\documentclass{llncs}

\usepackage{amsmath}
\usepackage{amssymb}
\usepackage{enumerate}
\usepackage{varioref}
\usepackage{ccfonts}
\usepackage{charter,eulervm}
\usepackage{boxedminipage}
\usepackage{color}

\usepackage[T1]{fontenc}

\DeclareMathOperator{\NP}{\mathrm{NP}}

\DeclareMathOperator{\diam}{\mathrm{diam}}
\DeclareMathOperator{\AT}{\mathrm{AT}}

\newcommand{\es}{\varnothing}

\title{\sc {Convexities in Some Special Graph Classes ---
  New Results in AT-free Graphs and Beyond}} 

\author{
Wing-Kai~Hon\inst{1} 
\and 
Ton~Kloks\inst{2} 
\and 
Hsiang-Hsuan~Liu\inst{1} 
\and
Hung-Lung~Wang\inst{2}
\and 
Yue-Li~Wang\inst{3}}

\institute{
Department of Computer Science\\
National Tsing Hua University, Taiwan\\
{\tt hhliu@cs.nthu.edu.tw}
\and
Institute of Information and Decision Sciences\\
National Taipei University of Business, Taiwan\\
{\tt hlwang@ntub.edu.tw}
\and 
Department of Information Management\\
National Taiwan University of Science and Technology, Taiwan\\
{\tt ylwang@cs.ntust.edu.tw}} 

\begin{document}

\maketitle 

\begin{abstract}
We study convexity properties of graphs.  In this paper we present a 
     linear-time algorithm for the geodetic number in tree-cographs.  Settling a 
     10-year-old conjecture,  we prove that the Steiner number is at least the 
     geodetic number in AT-free graphs. Computing a maximal and proper monophonic set in 
     $\AT$-free graphs is NP-complete. We present polynomial algorithms for the 
     monophonic number in permutation graphs and the geodetic number in $P_4$-
     sparse graphs.
\end{abstract}

\section{Introduction}

The geodetic number of a graph was introduced by 
Buckley, Harary and Quintas and 
by Harary, Loukakis and Tsouros~\cite{kn:buckley2,kn:harary} 
(see also~\cite{kn:pelayo} for a recent survey and~\cite{kn:hansen} 
for pointing out some errors in~\cite{kn:harary}). 
It is defined as follows. 
A geodesic in a graph is a shortest path between two vertices, that is, 
a path that connects the two vertices with a minimal number of edges. 
Let $G=(V,E)$ be a graph. We write $n=|V|$ and $m=|E|$.   
For a set $S \subseteq V$ let 
\begin{equation}
\label{eqn1}
I(S)=\{\; z\;|\; \exists_{x,y \in S} \; 
\text{$z$ lies on a $x,y$-geodesic}\;\}.
\end{equation} 
The set $I(S)$ is called the geodetic closure, or also, the interval of $S$. 
A set $S$ is called convex if $I(S)=S$. 
A set $S$ is \underline{geodetic} if 
$I(S)=V$. 
The geodetic number $g(G)$ of $G$ is defined as the minimal cardinality 
of a geodetic set. 

\bigskip 

For a set $S$, the interval $I(S)$ can be computed in 
$O(|S|\cdot m)$ time~\cite{kn:dourado}. 
The computation of the geodetic number is $\NP$-complete, even when 
restricted to chordal graphs, chordal bipartite graphs, or cobipartite 
graph~%
~\cite{kn:atici,kn:dourado3,kn:ekim}. 
It is polynomial for cographs and splitgraphs~\cite{kn:dourado3},  
for unit interval graphs, bipartite permutation graphs and 
block-cactus graphs~\cite{kn:ekim} and for 
Ptolemaic graphs~\cite{kn:farber}. 
It can be seen that the geodetic number problem  
can be formulated in monadic second-order logic 
(MSOL, see, eg,~\cite{kn:kante}). 
 
\bigskip 

Let $G=(V,E)$ be a graph and $W \subseteq V$. A \textit{Steiner $W$-tree} 
is a connected subgraph of $G$ with a minimal number 
of edges that contains all vertices of $W$. The Steiner distance of 
$W$ is the size of a Steiner $W$-tree. 
The Steiner interval $S(W)$ is the set of all vertices that are 
in some Steiner $W$-tree. If $S(W)=V$ then $W$ is a Steiner set. 
The Steiner number $s(G)$ is defined as the minimal cardinality  
of a Steiner set~\cite{kn:pelayo}. 

\bigskip 

For graphs in general there is no order relation between the 
Steiner number and the geodetic number~\cite{kn:hernando,kn:pelayo}.   
For distance-hereditary graphs it was shown that every 
Steiner set is geodetic, that is, 
$g(G) \leq s(G)$.%
The same was proved for interval graphs~\cite{kn:hernando,kn:pelayo}. 
In their paper the authors posed the question whether the same holds 
true for $\AT$-free graphs.  We answer the question in Section~\ref{section:at_free}. 

\bigskip

The main result of this paper is on the geodetic number and the Steiner number. We start discussion on the geodetic number with a simple graph structure, the \textit{tree-cograph}, defined in Section~\ref{section:tree_cograph}, and show that $g(G)$ can be computed in linear time when $G$ is a tree-cograph. Next, in Section~\ref{section:at_free}, we investigate the relation between the geodetic number and the Steiner number of an $\AT$-free graph. We show that in an $\AT$-free graph $G$,  every Steiner set is geodesic, i.e., $g(G)\leq s(G)$. This answers a question posed by Hernando, Jiang, Mora, Pelayo, and Seara in 2005~\cite{kn:hernando}. A closely related concept, the \textit{monophonic set} (defined in Section~\ref{section:monophonic}), is also investigated, and we show that computing a maximal and proper monophonic set is NP-complete, even for $\AT$-free graphs.  

\bigskip

Because of space limitations we relocate some results to appendices, including  the 2-geodetic number for a tree-cograph (Appendix~\ref{section:2_geodetic_number}), the geodetic number of a $P_4$-sparse graph (Appendix~\ref{section P4sparse}), and the \textit{monophonic number}, which is the minimal cardinality of a monophonic set, of a permutation graph (Appendix~\ref{section:monophonic_permutation}).

\section{The geodetic number for tree-cographs}
\label{section:tree_cograph}

Tree-cographs were introduced by Tinhofer in~\cite{kn:tinhofer}. 
They are defined as follows. 

\begin{definition}
A graph $G$ is a tree-cograph if one of the following holds. 
\begin{enumerate}[\rm 1.]
\item $G$ is disconnected.
\item $\Bar{G}$ is disconnected. 
\item $G$ or $\Bar{G}$ is a tree. 
\end{enumerate}
\end{definition}

\bigskip 

To compute the geodetic number for tree-cographs, we need an algorithm to compute the number for trees. In the following, we show how the computation can be done in linear time.

\begin{lemma}
\label{lm simplicial}
Let $G$ be a graph and let $x$ be a simplicial vertex. 
Then $x$ is an element of every geodetic set in $G$. 
\end{lemma}
\begin{proof}
A simplicial occurs in a geodesic only as an endvertex. 
\qed\end{proof}

\medskip 

\begin{lemma}
\label{lm tree}
Let $T$ be a tree. The set of leaves of $T$ forms a minimum 
geodetic set. 
\end{lemma}
\begin{proof}
By Lemma~\ref{lm simplicial}, any minimum geodetic set contains 
all the leaves. Since any other, that is, internal vertex, of 
$T$ lies on a geodesic between two leaves, the set of leaves 
forms a geodetic set. 
\qed\end{proof}

\medskip 

\begin{lemma}
\label{lm cotree}
Let $T$ be a tree with $n$ vertices. 
\begin{enumerate}[\rm (a)]
\item If $\diam(T) \leq 2$ then $g(\Bar{T})=n$. 
\item If $\diam(T) =3$ then $g(\Bar{T})=2$. 
\item If $T$ has a vertex of degree two and $\diam(T)>3$ 
then $g(\Bar{T})=3$. 
\item Otherwise, $g(\Bar{T})=4$. 
\end{enumerate}
\end{lemma}
\begin{proof}
Assume that $n > 3$ and that $\diam(T) \geq 3$. 
If $\diam(T)=3$ then $g(\Bar{T})=2$.  
If $T$ has a $P_3$ with the middle vertex of degree 2, then 
$g(\Bar{T}) \leq 3$. 
In this case, if $\diam(T)=3$ 
then $g(\Bar{T})=2$, and $g(\Bar{T})=3$ otherwise.  

\medskip 

\noindent
Henceforth assume that $T$ has no vertex of degree 2 and $\diam(T)>3$. 
It is easy to see that, if $\diam(T) \geq 5$ then $g(\Bar{T})=4$ and 
if $\diam(T)=4$ then a geodetic set needs all the vertices of a $P_5$ 
except some endpoint.  

\medskip 

\noindent
This proves the lemma.  
\qed\end{proof}

\bigskip 
Instead of reducing to MSOL, we compute the geodetic number with a relatively simple and efficient algorithm. The proposed linear-time algorithm is based on a parameter called \textit{2-geodetic number} of a graph, defined as follows.

\medskip

\begin{definition}
Let $G$ be a graph. A geodetic set $S \subseteq V$ is $2$-geodetic%
\footnote{The $2$-geodetic convexity should not be 
confused with the $P_3$-convexity,  
studied, e.g., in~\cite{kn:centeno}.} 
if 
every vertex $x$ of $V \setminus S$ has two nonadjacent neighbors 
in $S$, that is, there are vertices $a,b \in N(x) \cap S$ 
such that $[a,x,b]$ induces a $P_3$.  
The $2$-geodetic number, $g_2(G)$, of $G$ is the minimal cardinality of a 
$2$-geodetic set in $G$. 
\end{definition}

\bigskip

\begin{lemma}
\label{lm 2-geod tree}
There exists a linear-time algorithm to compute 
the $2$-geodetic number for trees. 
\end{lemma}
\begin{proof}
The following is a dynamic programming
algorithm to compute $g_2(T)$ for tree $T$ rooted at an 
arbitrary vertex $r$.
The subtree of $T$ with root $v$ is denoted by $T_v$ and $C(v)$ 
denotes the set containing all children of $v$. 
Let $S$ be a minimal 2-geodetic set. 
Let $\alpha_v$, $\beta_v$, and $\gamma_v$ be 
the numbers of vertices in $S\cap T_v$ 
when $v\in S$, $v\notin S$ and $|S\cap C(v)|=1$, and 
$v\notin S$ and $|S\cap C(v)|\geqslant 2$, respectively.
The values of $\alpha_v$, $\beta_v$, and $\gamma_v$ can be 
computed as follows.
\begin{multline}
\alpha_v=
\begin{cases}
1 & \text{if $v$ is a leaf}\\
1+\sum_{x\in C(v)}\min\{\alpha_x,\beta_x,\gamma_x\} & \text{otherwise.}
\end{cases}\\
\shoveleft{\beta_v=
\begin{cases}
n & \text{if $v$ is a leaf}\\
\omega_{v,1}+
\sum_{x\in C(v)}\min\{\beta_x,\gamma_x\} & \text{otherwise.}
\end{cases}}\\
\shoveleft{\gamma_v=
\begin{cases}
n & \text{if $v$ is a leaf}\\
\omega_{v,2}+\omega_{v,3}+
\sum_{x\in C(v)}\min\{\alpha_x,\beta_x,\gamma_x\} & 
\text{otherwise.}
\end{cases}}\\
\shoveleft{\text{Here}\quad  
\omega_{v,1}=\min\limits_{x\in C(v)}
\{\alpha_x-\min\{\beta_x,\gamma_x\}\}}, 
\quad \omega_{v,2}  = \min\limits_{x\in C(v)}
\{\alpha_x-\min\{\alpha_x,\beta_x,\gamma_x\}\},\\
\shoveleft{ 
\omega_{v,3} = \min\limits_{x\in C(v)\setminus\{u\}}
\{\alpha_x-\min\{\alpha_x,\beta_x,\gamma_x\}\}},
\quad \text{and $u$ is a vertex with} \\
\omega_{v,2}=\alpha_u-\min\{\alpha_u,\beta_u,\gamma_u\}.
\end{multline}
Finally, we have $g_2(T)=\min\{\alpha_r,\gamma_r\}$. 
It is clear that $g_2(T)$ can be computed in linear time. 

\medskip 

\noindent
This completes the proof.
\qed\end{proof}

\bigskip 

\begin{lemma}
\label{lm 2-geod comp tree}
Assume $T$ is a tree. Then 
\[g_2(\Bar{T})= 
\begin{cases}
3 & \text{if $\diam(T)=3$ and there is a vertex of degree two}\\ 
4 & \text{if $\diam(T)=3$ and no vertex has degree two}\\
g(\Bar{T}) & \text{otherwise.}
\end{cases}\]
\end{lemma}
\begin{proof}
Checking the cases in the proof of Lemma~\ref{lm cotree} 
it follows that each of the geodetic sets in $\Bar{T}$ is actually a 
$2$-geodetic set unless $\diam(T)=3$. 
\qed\end{proof}

\bigskip
\noindent

\begin{remark}
The 2-geodetic number of a tree-cograph can be computed in linear time
  (see Appendix A).  
 \end{remark}
  
  Based on this result, we have the following theorem.
 
\medskip  
  
\begin{theorem}
There exists a linear-time algorithm that computes the geodetic 
number of tree-cographs. 
\end{theorem}
\begin{proof}
Assume that $G$ is a tree-cograph. If $G$ or $\Bar{G}$ is a tree 
then the claim follows from Lemmas~\ref{lm tree} and~\ref{lm cotree}. 

\medskip 

\noindent 
Assume that $G$ is disconnected. Let $C_1,\dots,C_t$ 
be the components. 
Then 
\[g(G)=\sum_{i=1}^t g(C_i),\] 
where we write $g(C_i)$ instead of $g(G[C_i])$ for convenience. 

\medskip 

\noindent 
Assume that $\Bar{G}$ is disconnected, and let $C_1,\dots,C_t$ 
be the components of $\Bar{G}$. Assume that the components 
are ordered, such that 
\[|C_i| \geq 2 \quad\text{if and only if}\quad 1 \leq i \leq k.\] 
We claim that  
\begin{equation}
\label{eqn2}
g(G)=
\begin{cases}
n & \text{if $k=0$}\\
g_2(C_1) & \text{if $k=1$}\\
\min\;\{\;4,\;g_2(C_i)\; \mid\; 1 \leq i \leq k\;\} & \text{if $k \geq 2$.}
\end{cases}
\end{equation}
To prove the claim, first observe that, when $k=0$, $G$ is a 
clique and $g(G)=n$. Assume that $k=1$. 
Let $D$ be a minimum $2$-geodetic set in $G[C_1]$. Then it contains 
two nonadjacent vertices of $C_1$. Then $D$ is a geodetic set in $G$. 
Now let $D^{\prime}$ be a minimum geodetic set in $G$. 
It contains two nonadjacent vertices, which must be in $C_1$. 
Then $D^{\prime} \cap C_1$ is a $2$-geodetic set in $G[C_1]$. 

\medskip 

\noindent 
Assume that $k \geq 2$. 
Any 4 vertices, of which two are a nonadjacent 
pair in $C_1$ and another two are a nonadjacent pair in $C_2$, 
form a geodetic 
set in $G$. Thus $g(G) \leq 4$. Assume that $g(G)=2$. 
Then a minimum geodetic set must consist of two nonadjacent vertices, 
which 
must be contained in one $C_i$. 
Those form a $2$-geodetic set in $G[C_i]$ 
and so Formula~\eqref{eqn2} holds true. Assume that $g(G)=3$. 
It cannot be that two are a nonadjacent pair in one component and 
the third vertex is in another component, since then the two would 
be a geodetic set, contradicting the minimality. Therefore, the three 
must be a contained in one component. This proves Formula~\eqref{eqn2}.

\medskip 

\noindent 
This proves the theorem.  
\qed\end{proof}

\medskip 


\section{Steiner sets in $\AT$-free graphs}
\label{section:at_free}

Asteroidal triples were introduced by Lekkerkerker and Boland to 
identify those chordal graphs that are interval graphs~\cite{kn:lekkerkerker} 
(see also, eg,~\cite{kn:kloks}). 

\begin{center}
\begin{boxedminipage}{12cm}
An asteroidal triple, $\AT$ for short, 
is a set of 3 vertices $\{x,y,z\}$ such that there 
exists a path connecting any pair of them that avoids the closed 
neighborhood of the third. 
\end{boxedminipage}
\end{center}
A graph is \underline{$\AT$-free} if it has no asteroidal triple. 
Well-known examples of $\AT$-free graphs are cocomparability graphs, 
that is, the complements of comparability graphs. However, 
$\AT$-free graphs need not be perfect; the $C_5$ is $\AT$-free. 

%
%
%

\bigskip 

\begin{definition}
A dominating pair in a connected graph is a pair of vertices 
such that every path between them is a dominating set. 
\end{definition}

The following result was proved in~\cite{kn:broersma}. 


\begin{lemma}
\label{lm caterpillar}
Let $G$ be an $\AT$-free graph and let $W \subseteq V$. 
Let $T$ be a Steiner $W$-tree. There exists a Steiner $W$-tree 
$T^{\prime}$ with $V(T)=V(T^{\prime})$ and such that 
$T^{\prime}$ is a caterpillar.
\end{lemma}

\bigskip 

\begin{theorem}
Let $G$ be $\AT$-free. 
Let $W$ be a Steiner set of $G$. 
Then $W$ is also a geodetic set. So, for $\AT$-free graphs $G$ holds 
that $g(G) \leq s(G)$.  
\end{theorem}
\begin{proof}
Let $z \in V \setminus W$. We prove that there exist 
vertices $w^{\prime},w^{\prime\prime} \in W$ such that 
$z\in I(\{w^{\prime},w^{\prime\prime}\})$. 

\medskip 

\noindent
By assumption, there exists a Steiner $W$-tree $T$ that contains $z$, 
and by Lemma~\ref{lm caterpillar} we may assume that $T$ is a 
caterpillar. 
Let $P$ denote the backbone of the caterpillar, and let $x$ and $y$ be the endpoints of $P$. 
We may also assume that $P$ is chordless, and $x$ and $y$ are leaves of $T$.  
Notice that all leaves of $T$ are elements of $W$.

\medskip 

\noindent 
Since $z\in V\setminus W$, we have $z \in V(P)$. 
Let $w^{\prime}$ and $w^{\prime\prime}$ be two vertices on 
opposite sides of $z$ such that 
\[z \in P[w^{\prime},w^{\prime\prime}] \quad \text{and}\quad 
P[w^{\prime},w^{\prime \prime}] \cap W = \{\;w^{\prime},\;w^{\prime\prime}\;\}.\] 
If $P[w^{\prime},w^{\prime \prime}]$ is not a shortest $w^{\prime},w^{\prime\prime}$-path in $G$, we claim that there is a cycle, which is the union of two paths $Q$ and $Q^{\prime}$ in $G$, such that
\begin{itemize}
\item $Q=P[x^{\prime}, y^{\prime}]$
\item $z\in V(Q)$
\item $Q^{\prime}$ is a shortest $x^{\prime}, y^{\prime}$-path in $G$.
\item $V(Q^{\prime})\cap V(P)=\{x^{\prime}, y^{\prime}\}$
\item $|V(Q^{\prime})|<|V(Q)|$
\end{itemize}  

\medskip 

\noindent 
Let $P^{\prime}$ be a shortest $w^{\prime},w^{\prime\prime}$-path in $G$, and let $V(P)\cap V(P^{\prime})=\{a_1, a_2,\dots,a_k\}$, where elements are ordered according to the traversal of $P^{\prime}$ from $w^{\prime}$ to $w^{\prime\prime}$. Clearly, $k\geq 2$, and for $i\in[k-1]$ we have that $P^{\prime}[a_i, a_{i+1}]$ is a chordless path in $G$. Let $i^*$ be the least integer such that $z\in V(P[a_{i^*},a_{i^*+1} ])$. It follows that  $P^{\prime}[a_{i^*}, a_{i^*+1}]$ is shortest and shorter than $P[a_{i^*}, a_{i^*+1}]$.\footnote{Otherwise, we can shorten $P$ by replacing $P^{\prime}[a_{j}, a_{j^{\prime}}]$ by $P[a_{j}, a_{j^{\prime}}]$, where $j\leq i^*$, $j^{\prime}\geq i^*+1$, and $P[a_{j}, a_{j^{\prime}}]$ contains no $a_i$ other than $a_{j}$ and $a_{j^{\prime}}$.}
By letting $x^{\prime}=a_{i^*}$, $y^{\prime}=a_{i^*+1}$, and $Q^{\prime}=P^{\prime}[x^{\prime}, y^{\prime}]$, we have the requested cycle. 


\medskip 

\noindent 
In the following, we develop the case where $V(P[x^{\prime}, y^{\prime}])\cap W\subseteq \{x^{\prime}, y^{\prime}\}$. The case where $P[x^{\prime}, y^{\prime}]$ contains more vertices of $W$ can be manipulated in a similar manner.\footnote{In this case, we choose $Q$ as the maximal subpath of $P[x^{\prime}, y^{\prime}]$ containing no vertices in $W$ except the endpoints, and $Q^{\prime}$ be the remainder of the cycle.}
Next, we claim that all vertices of $Q$ and all leaves attached to these vertices, except those attached to only $x_1$ or $y_1$, are adjacent to some vertex of $Q^{\prime}$, where $N(x^{\prime})\cap V(Q)=\{x_1\}$ and $N(y^{\prime})\cap V(Q)=\{y_1\}$.
Formally, let $L$ be the set of leaves of $T$. The claim states that
\begin{equation}\label{eqn:10year}
\forall_{v\in V(Q)\cup L^{\prime}}\quad N[v]\cap V(Q^{\prime})\neq\es,
\end{equation}
where $L^{\prime}=\{\,v\in L \mid N(v)\cap (V(Q)\setminus\{x_1, y_1\})\neq\es\,\}$. 
To see that, assume that $u\in V(Q)\cup L^{\prime}$ and $N(u) \cap V(Q^{\prime})=\es$. Then $\{x^{\prime}, y^{\prime},u\}$ forms an asteroidal triple, since there is a $u,x^{\prime}$-path that avoids $N[y^{\prime}]$ via $Q$, a $u,y^{\prime}$-path that avoids $N[x^{\prime}]$ via $Q$ and an $x^{\prime}, y^{\prime}$-path that avoids $N[u]$ 
via $Q^{\prime}$. Thus, condition~\eqref{eqn:10year} holds. Moreover, let $q$ and $q^{\prime}$ be the length of $Q$ and $Q^{\prime}$, respectively. We may assume that $q\leq q^{\prime}+2$ since otherwise we can replace $Q$ with $Q^{\prime}$ concatenated with $x_1$ and $y_1$ to get a smaller tree containing $W$.

\medskip
\noindent
Consider the following cases.

\begin{enumerate}[(i)]
\item $|\{x^{\prime}, y^{\prime}\}\cap \{x, y\}|=0$: For each leaf $u$ of $T$ adjacent to only $x_1$ or $y_1$, there is an asteroidal triple involving these endpoints, i.e. $(x, u, y)$, unless $N(u)\cap Q^{\prime}\neq\es$. 
Thus, with~\eqref{eqn:10year}, we can replace $Q$ with $Q^{\prime}$ to form the backbone of a caterpillar containing $W$. This contradicts the minimality of $T$.

\medskip

\item $|\{x^{\prime}, y^{\prime}\}\cap \{x, y\}|=1$: Assume that $x^{\prime}=x$. Similar to the previous case, for each leaf $u$ attached only at $y_1$, we have $(x, u, y)$ as an asteroidal triple, unless $u$ is adjacent to some vertex of $Q^{\prime}$. For $q^{\prime}=q-2$, or there is no leaf attached at $x_1$, this leads to a contradiction, as in the previous case. Thus, we assume that there is a leaf $w_1$ attached only at $x_1$ and $q^{\prime}=q-1$. Notice that $N(w_1)\cap (V(P)\cup V(Q^{\prime}))=\{x_1\}$. If $z=x_1$, then clearly $z$ is on a shortest $x,w_1$-path. If $z\neq x_1$, we claim that $z$ is on a shortest $w_1, w^{\prime\prime}$-path. 

\medskip

\noindent
To see this,  suppose to the contrary that 
\begin{equation}\label{eq:single_end_distance}
d(w_1, w^{\prime\prime})<q+d(y^{\prime}, w^{\prime\prime}).
\end{equation}
If there is a shortest $w_1, w^{\prime\prime}$-path passing through a neighbor $v$ of $x$, then 
\[d(w_1, w^{\prime\prime})\geq 2+q^{\prime}-1+d(y^{\prime}, w^{\prime\prime})=q+d(y^{\prime}, w^{\prime\prime}),\]
since this path is of the form $w_1\leadsto u\leadsto v\leadsto y\leadsto w^{\prime\prime}$, where $u$ is a neighbor of $w_1$. This contradicts~\eqref{eq:single_end_distance}, and we may assume that no shortest $w_1, w^{\prime\prime}$-path contains a neighbor of $x$. However, this leads to the existence of an asteroidal triple $(w_1, x, y)$, again a contradiction. This proves the claim. 

\medskip

\item $|\{x^{\prime}, y^{\prime}\}\cap \{x, y\}|=2$: We can shorten $T$ by replacing the backbone if the leaves attached at neither $x_1$ nor $y_1$. If $q^{\prime}=q-1$, or exactly one of $x_1$ and $y_1$ has a neighbor in $L\setminus L^{\prime}$, then as in (ii), we have that $z$ is on a shortest path between two vertices in $W$. The only case left of interest is when $q^{\prime}=q-2$, and there are leaves, $w_1$ and $w_2$, attached at $x_1$ and $y_1$, respectively. If $z\in\{x_1, y_1\}$, then $z$ is on the shortest $x, w_1$-path or $y, w_2$-path. Otherwise, we claim that $z$ is on a shortest $w_1,w_2$-path. 

\medskip
\noindent
To see this, suppose to the contrary that 
\begin{equation}\label{eq:double_side_distance}
d(w_1, w_2)<q.
\end{equation}
If there is a shortest $w_1,w_2$-path passing through two vertices which are neighbors of $x$ and $y$,  respectively, then we have
\[d(w_1, w_2)\geq 2+q^{\prime}-2+2=q,\] 
since such a path is of the form $w_1\leadsto u_1\leadsto v_1\leadsto v_2\leadsto u_2\leadsto w_2$, where $u_1\in N(w_1)$, $v_1\in N(x)$, $v_2\in N(y)$, and $u_2\in N(w_2)$. This contradicts~\eqref{eq:double_side_distance}. However, if each shortest $w_1,w_2$-path contains no neighbor of $x$ or no neighbor of $y$, then $(x, w_1, y)$ or $(y, w_2, x)$ is an asteroidal triple, again a contradiction.

\end{enumerate}

\noindent
This proves the theorem. 
\qed\end{proof}
 
\bigskip 
\noindent
Although $g(G)\leq s(G)$ when $G$ is $\AT$-free, the equality is not guaranteed to hold even for subclasses like unit-interval graphs, as shown in Theorem~\ref{thm:unit_interval}.

\medskip

\begin{theorem}\label{thm:unit_interval}
For unit-interval graphs $G$, the geodetic number $g(G)$ and the Steiner 
number $s(G)$ are, in general, not equal.
Moreover, the difference between the two numbers can be arbitrarily 
large.
\end{theorem}
\begin{proof}
Consider a set of seven unit-length intervals as depicted in 
Figure~\ref{fig ui counter example 2}.  Let $G$ be the interval graph 
corresponding to these intervals; for ease of discussion, we abuse 
the notation slightly to refer to a vertex in $G$ by the label of its 
corresponding interval.  
It is easy to check that $g(G) \leq 4$, as 
$\{I_1, I_3, I_4, I_7\}$ 
forms a geodetic set.  
We shall show that  $s(G) > 4$, so that the first 
statement of the theorem follows.  

\medskip 

\noindent
First, $I_1$ and $I_3$ are simplicial vertices, 
so that any minimal Steiner set must 
include $I_1$ and $I_3$.  
On the other hand, $I_1$ and $I_3$ alone do not form a Steiner set, 
since $\{I_1,I_2,I_3\}$ forms the only Steiner $\{I_1,I_3\}$-tree 
and, eg, $I_4$ does not lie on that. 

\medskip 

\noindent
Next, if a Steiner set includes $I_2$, it has 
to include all the remaining vertices 
(because the subgraph induced by any superset 
of $\{I_1,I_2,I_3\}$ is connected,  
and so, any remaining vertex would not be included 
in a Steiner tree). 
Similarly, if a Steiner set includes $\{I_1,I_3,I_4,I_7\}$ then 
it has to include all the remaining vertices.

\medskip 

\noindent
Thus, either a Steiner set has size $9$, or it must 
include $I_4$ but not $I_7$ (or vice versa);  further, 
if a Steiner set includes $I_4$, then it has to 
include $I_5$ and $I_6$ also.  Thus, $s(G) \geq 5$.
This completes the proof of the first statement.
To show the second statement, it suffices to duplicate arbitrarily many 
disjoint (i.e., non-overlapping) copies of 
the set of intervals in our example.
\qed
\end{proof}

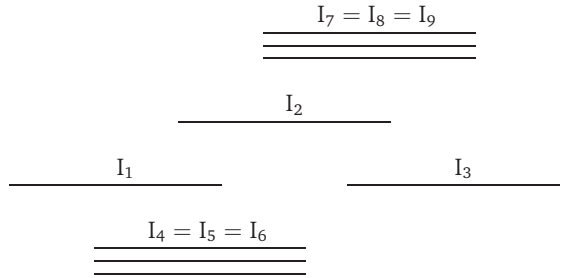
\begin{figure}[t]
\setlength{\unitlength}{8pt}
\begin{center}
\begin{picture}(26,10)
\put(0,3){\line(1,0){10}}
\put(5,3.5){$I_1$}
\put(8,6){\line(1,0){10}}
\put(13,6.5){$I_2$}
\put(16,3){\line(1,0){10}}
\put(21,3.5){$I_3$}
\put(4,0){\line(1,0){10}}
\put(4,-0.6){\line(1,0){10}}
\put(4,-1.2){\line(1,0){10}}
\put(6.5,0.5){$I_4 = I_5 = I_6$}
\put(12,9){\line(1,0){10}}
\put(12,9.6){\line(1,0){10}}
\put(12,10.2){\line(1,0){10}}
\put(14.5,10.7){$I_7 = I_8 = I_9$}
\end{picture}
\end{center}
\caption{An example showing that $g(G) \neq s(G)$ for 
unit-interval graphs in general.  
Our example graph $G$ consists of nine unit intervals:  
$I_1 = [0.0, 1.0]$, $I_2 = [0.8, 1.8]$, $I_3 = [0.6, 2.6]$, 
$I_4 = I_5 = I_6 = [0.4, 1.4]$, $I_7 = I_8 = I_9 = [1.2, 2.2]$, 
and it follows that $g(G) \leq 4 < s(G)$.}
\label{fig ui counter example 2}
\end{figure}

%

\section{Monophonic sets in $\AT$-free graphs}
\label{section:monophonic}


Let $G=(V,E)$ be a graph. 
For a set $S \subseteq V$ the monophonic closure is the set 
\[J(S)=\{\;z\;|\; \exists_{x,y \in S} \;\text{$z$ lies on a chordless 
$x,y$-path}\;\}.\] 
For pairs of vertices $x$ and $y$ we write $J(x,y)$ instead of $J(\{x,y\})$.  
A set $S$ is \underline{monophonic} if $J(S)=V$. 
The monophonic number $m(G)$ of $G$ is the minimal cardinality of a 
monophonic set in $G$. 
Some complexity results on monophonic convexity appear in~\cite{kn:dourado4}. 
Computing the monophonic number of a graph in general is $\NP$-complete. 
In cographs, as in distance-hereditary graphs, 
chordless paths are geodesics, and so, 
$m(G)=g(G)$. 
The monophonic hull number turns out to be polynomial; actually, if a graph 
is not a clique and contains no clique separator, then it is called 
an atom, and in atoms, every pair of 
nonadjacent vertices forms a monophonic hull set~\cite{kn:duchet} 
and~\cite[Theorem~5.1]{kn:dourado4} 

\bigskip

\begin{theorem}[See~\cite{kn:hernando}]
Let $G$ be connected. Then every Steiner set in $G$ is monophonic. 
Consequently, 
\[m(G) \leq s(G).\] 
\end{theorem}

\bigskip 

\begin{remark}
Let $c_m(G)$ denote the cardinality of a maximum, proper,  
monophonically convex subset of $G$. 

\begin{theorem}
Computing $c_m(G)$ is $\NP$-complete for $\AT$-free graphs. 
\end{theorem}
\begin{proof}
Computing the clique number $\omega$ is $\NP$-complete for 
$\AT$-free graphs~\cite{kn:broersma}. 
Let $H$ be $\AT$-free. We may assume that $\omega(H) < |V(H)|-1$. 
Let $G$ be the graph obtained from $H$ 
as follows. Add two nonadjacent vertices $u$ and $v$ and make 
each adjacent to all vertices of $H$. Notice that $G$ is 
$\AT$-free. 
As in~\cite{kn:dourado4}, it is easy to see that $c_m(G) \geq k+1$ 
if and only if $\omega(H) \geq k$. 
\qed\end{proof}
\end{remark}

\medskip 

\begin{remark}
Basically, an application of Dirac's theorem shows that  
that the monophonic number of a chordal graph 
equals the number of simplicial vertices~\cite{kn:chvatal,kn:farber}.
%
\end{remark}

\bigskip 


\newpage
\appendix 

\section{The 2-geodetic number for tree cographs}
\label{section:2_geodetic_number}

\begin{lemma}
\label{lm 2-geod tree-cograph}
There exists a linear-time algorithm to compute the 
2-geodetic number of tree-cographs. 
\end{lemma}
\begin{proof}
Let $G$ be a tree-cograph. 
By Lemmas~\ref{lm 2-geod tree} and~\ref{lm 2-geod comp tree} 
we may assume that $G$ is not a tree or the 
complement of a tree. 

\medskip 

\noindent 
Assume that $G$ is disconnected and let $G$ be a union of two 
tree-cographs $G_1$ and $G_2$. 
In that case 
\[g_2(G)=g_2(G_1)+g_2(G_2).\] 

\medskip 

\noindent
Now assume that $\Bar{G}$ is disconnected. 
Let $C_1,\dots,C_t$ be the components of $\Bar{G}$ and let them 
be ordered such that 
\[|C_i| \geq 2 \quad \text{if and only if}\quad 1 \leq i \leq k.\] 
If $k=0$ then $G$ is a clique and $g_2(G)=n$. 

\medskip 

\noindent 
Assume that $k=1$. Then 
\[g_2(G)=g_2(C_1),\]  
where we write $g_2(C_1)$ instead of $g_2(G[C_1])$, for convenience 
and clarity of notation. 
To see that, let $D$ be a $2$-geodetic set in $G$. Then $D \cap C_1$ 
is a $2$-geodetic set in $G[C_1]$. Conversely, let $D^{\prime}$ be a 
$2$-geodetic set in $G[C_1]$. Then $D^{\prime}$ contains two nonadjacent 
vertices since $G[C_1]$ is not a clique. Then $D^{\prime}$ is also a 
$2$-geodetic set in $G$. 

\medskip 

\noindent 
Assume that $k \geq 2$. 
We claim that 
\[g_2(G)=\min \; \{\; 4,\; g_2(C_i)\;|\; 1 \leq i \leq k\;\}.\] 
Notice that a set of 4 vertices, 2 nonadjacent in $C_1$ and 2 
nonadjacent in $C_2$, form a $2$-geodetic set in $G$. 
Thus $g_2(G) \leq 4$. 

\medskip 

\noindent
Assume that $g_2(G)=2$. Then a minimum $2$-geodetic set consists 
of two nonadjacent vertices, which must be contained in some $C_i$. 
Thus in that case, 
\begin{equation}
\label{eqntreeco2}
g_2(G)=\min \; \{\; g_2(C_i) \;|\; i \in \{1,\dots,k\}\;\}.
\end{equation}

\medskip 

\noindent
Assume that $g_2(G)=3$. Assume that two vertices of a minimum 
geodetic set are in $C_1$ and one is in $C_2$. Then the two 
in $C_1$ must be a $2$-geodetic set, contradicting the minimality. 
Therefore,~\eqref{eqntreeco2} holds also in this case. 

\medskip 

\noindent
This proves the lemma. 
\qed\end{proof}

\section{The geodetic number for $P_4$-sparse graphs}
\label{section P4sparse}

Ho\`ang introduced $P_4$-sparse graphs 
as follows.    

\begin{definition}
A graph is $P_4$-sparse if every set of 5 vertices induces at most one $P_4$. 
\end{definition}

Jamison and Olariu characterized $P_4$-sparse graphs 
using the notion of spiders. 

\begin{definition}
A graph $G$ is a \underline{thin spider} if 
its vertices are partitioned into 3 sets, 
$S$, $K$ and $R$, such that the following conditions hold. 
\begin{enumerate}[\rm 1.]
\item $S$ is an independent set and $K$ is a clique and 
\[|S|=|K| \geq 2.\] 
\item Every vertex of $R$ is adjacent to every vertex 
of $K$ and to no vertex of $S$. 
\item There is a bijection between $K$ and $S$ such that 
every vertex of $S$ 
is adjacent to a unique vertex in $K$. 
\end{enumerate}
A \underline{thick spider} is the complement of a thin spider. 
\end{definition}  
Notice that, possibly $R=\es$. 
The set $R$ is called the head, $K$ the body and 
$S$ the set of feet of the spider. For a thick spider we 
switch the notation 
$K$ and $S$ for the feet and the body when taking the 
complement, so that in both cases the head 
$R$ is adjacent to the body $K$.   

\bigskip 

The following characterization of $P_4$-sparse graphs. 

\begin{theorem}
A graph $G$ is $P_4$-sparse if and only if for every induced subgraph $H$ of 
$G$ one of the following conditions is satisfied. 
\begin{enumerate}[\rm (a)]
\item $H$ is disconnected. 
\item $\Bar{H}$ is disconnected. 
\item $H$ is isomorphic to a spider. 
\end{enumerate}
\end{theorem}

\bigskip 

\begin{lemma}
Let $G$ be $P_4$-sparse. There exists a linear-time algorithm to compute 
the $2$-geodetic number $g_2(G)$. 
\end{lemma}
\begin{proof}
The proof for the cases where 
$G$ or $\Bar{G}$ is disconnected is 
similar to the proof of Lemma~\ref{lm 2-geod tree-cograph}. 

\medskip 

\noindent 
Assume that $G$ is a thin spider, say with a head $R$, a body $K$ and 
a set of feet $S$. 
Since all feet are pendant vertices, they all have to be 
in any minimum $2$-geodetic set. So, 
\[R=\es \quad\Rightarrow \quad g_2(G)=|S|+1,\] 
since, choosing $S$ and one vertex in $K$ gives every 
other vertex of $K$ two nonadjacent neighbors. 
When $R \neq \es$, we have 
\[g_2(G)=|S|+g_2(R),\] 
where we write $g_2(R)$ instead of $g_2(G[R])$ for convenience. 
To see that, let $D$ be a minimum $2$-geodetic set in $G[R]$. 
Then $D \cup S$ is a $2$-geodetic set in $G$, since every vertex of $K$ 
is adjacent to one foot and one element of $D$. 
Now let $D^{\prime}$ be a minimum $2$-geodetic set in $G$. 
Then $D^{\prime} \cap R$ is a $2$-geodetic set in $G[R]$, 
since for a vertex in $R$ the two nonadjacent neighbors in $D^{\prime}$ 
must be in $D^{\prime} \cap R$. 

\medskip 

\noindent 
Assume that $G$ is a thick spider. Let the head be represented by $R$.  
Let the body, which is a clique joined to $R$, be represented by $K$ 
and let the set of feet be represented by $S$.  
Notice that, 
\[R=\es \quad\Rightarrow\quad 
g_2(G)=
\begin{cases}
3 & \text{if $|K|=|S|=2$} \\
|S| & \text{otherwise.}
\end{cases}
\] 
To see that, when $|K|=|S|=2$ the graph is a $P_4$, and $g_2(P_4)=3$. 
When $|S| > 2$, all vertices of $S$ must be in a minimum 
$2$-geodetic set, since they are simplicials. Since $|S| > 2$, 
every vertex of $K$ is adjacent to 2 nonadjacent feet.  

\medskip 

\noindent 
Finally assume that $G$ is a thick spider and assume that $R \neq \es$. 
In that case,      
\[g_2(G)=|S|+g_2(R).\] 
The argument is similar to the one we gave above. 

\medskip 

\noindent
This proves the lemma. 
\qed\end{proof}

\bigskip 
 
\begin{theorem}
There exists a linear-time algorithm to compute the 
geodetic number of a $P_4$-sparse graph. 
\end{theorem}
\begin{proof}
Let $G$ be a $P_4$-sparse graph. 
First assume that $G$ is disconnected. 
Let $C_1,\dots,C_t$ be the components of $G$. 
then 
\[g(G)=\sum_{i=1}^t \; g(G[C_i]).\] 

\medskip 

\noindent 
Assume that $\Bar{G}$ is disconnected. 
Let $C_1,\dots, C_t$ be the components of $\Bar{G}$. 
Assume that the components are ordered such that  
\[|C_i| \geq 2 \quad\text{if and only if} \quad 1 \leq i \leq k.\] 
Then we claim that  
\begin{equation}
\label{eqnsparse3}
g(G)=
\begin{cases}
n & \text{if $k=0$} \\
g_2(G[C_1]) & \text{if $k=1$} \\
\min \; \{\; 4, \; g_2(G[C_i])\;|\; 1 \leq i \leq k\;\} 
& \text{if $k \geq 2$.}
\end{cases}
\end{equation}
To prove the claim, notice that, when $k=0$, $G$ is a clique and 
then $g(G)=n$. Assume that $k=1$. Let $D$ be a minimum 
$2$-geodetic set in $G[C_1]$. 
Then this contains two nonadjacent vertices of $C_1$. Then 
$D$ is a geodetic set in $G$. For the converse, let 
$D^{\prime}$ be a minimum geodetic set in $G$. Then it contains 
two nonadjacent vertices, which must be in $C_1$. Then 
$D^{\prime} \cap C_1$ is a $2$-geodetic set in $G[C_1]$, 
since any geodesic has length two. Assume that 
$k \geq 2$. Notice that 4 vertices, of which two are nonadjacent 
elements of $C_2$ and two are nonadjacent elements of $C_2$, form a 
geodetic set. Therefore, $g(G) \leq 4$. Assume that 
$g(G)=2$. Then a minimum geodetic set contains two nonadjacent 
vertices which must be in some component, say $C_1$. Then those two 
vertices form a $2$-geodetic set in $G[C_1]$, and Equation~\eqref{eqnsparse3}  
holds true. Assume that $g(G)=3$. It cannot be that two are in $C_1$ 
and one is in $C_2$; that is, all three must be in one component of 
$\Bar{G}$. In that case, the three must form a $2$-geodetic set in that 
component. This proves the claim.     

\medskip 

\noindent
Assume that $G$ is a thin spider and let $R$, $K$ and $S$ be its head, body 
set of feet. 
By Lemma~\ref{lm simplicial} all the feet are in any geodetic set. 
When $R=\es$, then $S$ is a geodetic set, since all vertices 
of $K$ are in geodesics with endvertices in $S$. So we have 
\[R=\es \quad \Rightarrow \quad g(G)=|S|.\] 

\medskip 

\noindent
Assume that $R \neq \es$. 
Then 
\[g(G)=|S|+g_2(R),\] 
where we write $R$ instead of $G[R]$ for convenience.  
To see that, let $D$ be a minimum $2$-geodetic set in $G[R]$. 
Then $D \cup S$ is a geodetic set in $G$, since every vertex 
in $K$ has two nonadjacent neighbors in the set, one in $S$ and one in $R$. 
Let $D^{\prime}$ be a minimum geodetic set in $G$. 
Then $S \subseteq D^{\prime}$. 
Furthermore, $D^{\prime} \cap R$ is a $2$-geodetic set in $G[R]$, since 
any geodesic has length 2 unless $R$ is a clique. 

\medskip 

\noindent 
Assume that $G[R]$ is a thick spider. 
Then 
\[g(G)=|S|+g_2(R).\] 
The analysis is similar to the treatment of the thin spiders. 

\medskip 
  
\noindent
This proves the theorem. 
\qed\end{proof}

\section{The monophonic number in permutation graphs}
\label{section:monophonic_permutation}

A permutation diagram consists of two horizontal lines $L_1$ and $L_2$, 
one above the other. Each of the two lines has $n$ distinct, designated points on 
it, labeled in an arbitrary order, $1,2,\dots,n$. 
Each point on the topline is connected, {\em via\/} 
a straight 
linesegment, 
to the point on the bottom line that has the same label. 

\bigskip 

The companion of a permutation diagram is a permutation graph. 
The $n$ vertices of the graph are the $n$ linesegments in the diagram that 
connect the labeled points on the topline and bottom line. Any two vertices 
in the graph are adjacent precisely when the two linesegments intersect 
each other. 
In general, a graph is a permutation graph if it is represented by a permutation 
diagram. 

\bigskip 

Notice that, if a graph $G$ is a permutation graph then so is its complement $\Bar{G}$. 
By the transitivity of the left-to-right ordering 
of parallel linesegments, it follows that 
permutation graphs are comparability graphs. Baker 
et al. proved that these two properties characterize permutation graphs: 
a graph $G$ is a permutation graph if and only if 
$G$ and $\Bar{G}$ are comparability graphs
(see, e.g., also~\cite{kn:kloks}).

\begin{lemma}
Permutation graphs are $\AT$-free. 
\end{lemma}
\begin{proof}
Three, pairwise nonadjacent vertices in a permutation graph $G$,   
are represented by three parallel linesegments in its diagram. 
Some linesegment of a path that connects the outer two must 
intersect (or be equal to) the linesegment that is in the middle 
of the three.  
This implies that the path 
hits the closed neighborhood of the vertex in the middle. 
\qed\end{proof}

\bigskip 

An elegant notion of \underline{betweenness} for $\AT$-free graphs was introduced 
by Broersma et al. as follows. 

\begin{center}
\begin{boxedminipage}{12cm}
For two nonadjacent vertices $x$ and $y$ in a graph $G$ denote by $C_x(y)$ 
the component of $G-N[x]$ that contains $y$. 
\end{boxedminipage}
\end{center}

\begin{definition}
Let $x$ and $y$ be nonadjacent vertices. A vertex $z$ is 
\underline{between}  
$x$ and $y$ if 
\[z \in C_x(y) \cap C_y(x).\] 
\end{definition}

The following justification of this definition was proved in~\cite{kn:broersma}. 

\begin{theorem}
\label{thm between}
Let $G$ be $\AT$-free and let $z$ be a vertex between 
nonadjacent vertices $x$ and $y$. 
Then $x$ and $y$ are in different components of $G-N[z]$. 
\end{theorem}

\bigskip 
 
\begin{definition}
A pair of vertices $x$ and $y$ is extremal if the 
number of vertices between them is maximal. 
\end{definition}

\begin{lemma}
\label{lm extremal pair}
Let $G$ be $\AT$-free and let $\{x,y\}$ be an extremal pair. 
Let 
\[\Delta(x)=N(C_x(y)) \quad\text{and}\quad  
A(x)=V \setminus (A(x) \cup \Delta(x).\] 
Then every vertex of $A(x)$ is adjacent to every vertex of $\Delta(x)$. 
\end{lemma}
\begin{proof}
Let $C$ be the largest component of $G-N[x]$ and let $\Delta=N(C)$. 
Let $A=V \setminus (\Delta \cup C)$. Then every vertex of $A$ 
is adjacent to every vertex of $\Delta$, otherwise there would be a 
vertex $x^{\prime}$ for which the largest component of $G-N[x^{\prime}]$ 
properly contains $C$. 
\qed\end{proof}

Notice that, by Theorem~\ref{thm between} and Lemma~\ref{lm extremal pair}, 
every extremal pair is a dominating pair. 

\bigskip 

The following lemma gives the monophonic number for the 
graph induced by $A(x) \cup \Delta(x)$. 
For its proof we refer to~\cite{kn:paluga}. 

\begin{lemma}
\label{lm paluga}
Let $G_1$ and $G_2$ be two graphs and let $H=G_1 \otimes G_2$, 
that is, $H$ is the graph obtained from the union of 
$G_1$ and $G_2$ by adding 
all edges between pairs $x \in V(G_1)$ and $y \in V(G_2)$. 
Let $n_i=|V(G_i)|$, for $i \in \{1,2\}$. 
Then 
\[m(H)=
\begin{cases}
n_1+n_2 & \text{if $G_1 \simeq K_{n_1}$ and $G_2 \simeq K_{n_2}$}\\
m(G_2) & \text{if $G_1 \simeq K_{n_1}$ and $G_2 \not\simeq K_{n_2}$}\\
m(G_1) & \text{if $G_1 \not\simeq K_{n_1}$ and $G_2 \simeq K_{n_2}$}\\
\min\;\{\;4,\;m(G_1),\;m(G_2)\;\} & \text{if $G_1 \not\simeq K_{n_1}$ and 
$G_2 \not\simeq K_{n_2}$.}
\end{cases}\] 
\end{lemma}

\bigskip 
 
Let $\{x,y\}$ be an extremal pair. Then, by Lemma~\ref{lm extremal pair}, 
\[J(x,y) \subseteq \{x,y\} \cup \Delta(x) \cup \Delta(y) \cup 
\left(C_x(y) \cap C_y(x)\right).\] 
Unfortunately, we don't always have equality. 

\begin{lemma}
\label{lm J in C}
Let $C=C_x(y) \cap C_y(x)$ and let $z \in C$. 
Then 
\begin{multline}
z \notin J(x,y) \quad\Leftrightarrow\quad 
\forall_{a \in \Delta_z(x)} \; \forall_{b \in \Delta_z(y)}\;\;  
a=b \quad\text{or}\quad \{\;a,\;b\;\} \in E, \\
\text{where} \quad \Delta_z(x)=N(C_z(x)) \quad \text{and}\quad 
\Delta_z(y)=N(C_z(y)).
\end{multline}
\end{lemma}
\begin{proof}
When $z \in C$ and $z$ is on a chordless $x,y$-path $P$, then $z$ has neighbors 
$u$ and $v$ in $P$ which are not adjacent. By Theorem~\ref{thm between}, 
$N[z]$ separates $x$ and $y$ in different components, so    
$u \in N(C_z(x))$ and $v \in N(C_z(y))$ are not adjacent.     
For the converse, no chordless path can contain $z$ 
when every vertex of $C_z(x)$ is equal or adjacent to every 
vertex in $C_z(y)$, since any such path would have a chord. 
\qed\end{proof}
 
\bigskip 

Two minimal separators $S_1$ and $S_2$ are parallel if 
all vertices of $S_1 \setminus S_2$ are contained in one component of 
$G-S_2$ and all vertices of $S_2 \setminus S_1$ are contained 
in one component of $G-S_1$.
In such a case, 
let $C_2$ be the component of $G-S_1$ that contains $S_2\setminus S_1$  
and let $C_1$ be the component of $G-S_2$ that contains $S_1\setminus S_2$. 
The vertices between $S_1$ and $S_2$ are the vertices of $C_1 \cap C_2$. 
 
\bigskip 

Minimal separators in permutation graphs were analyzed 
in~\cite{kn:bodlaender}
with the notion of a scanline. Consider a permutation diagram. A scanline 
is a linesegment with one endpoint on the topline and one 
endpoint on the bottom line, such that neither of its endpoints coincides  
with a labeled point of the diagram. 
When $S$ is a minimal separator in a permutation graph 
then there is a scanline $s$ in the diagram such that the linesegments that 
cross $s$ are exactly the vertices of $S$. 

\begin{lemma}
\label{lm ordered neighborhoods}
Let $S$ be a minimal separator in a permutation graph $G$ and let 
$C$ be a component in $G-S$. The neighborhoods in $C$ of two nonadjacent 
vertices in $S$ are ordered by inclusion. 
\end{lemma}
\begin{proof}
Consider a scanline $s$. All the linesegments of a component 
occur on the same side of $s$; say the left side. 
Let $\tau$ and $\upsilon$ be nonadjacent 
vertices whose linesegments cross $s$. Say that $\tau$ and $\upsilon$ have their 
bottom endpoints on the left side of $s$ and say that the endpoint of 
$\upsilon$ on the bottom line 
is closer to $s$ than the endpoint of $\tau$. Then every linesegment 
of $C$ that crosses $\upsilon$ also crosses $\tau$, that is, 
\[N(\upsilon) \cap C \subseteq N(\tau) \cap C.\]
This proves the lemma.  
\qed\end{proof}

\begin{lemma}
Let $\{x,y\}$ be an extremal pair. Let $Q$ be a minimum monophonic set.  
Then 
\[Q \cap A(x)=\es \quad\Rightarrow\quad Q \cap \Delta(x) \neq \es.\] 
\end{lemma}
\begin{proof}
Let $Q$ be a minimum monophonic set and assume that $Q \cap A(x)=\es$. 
Assume that $x$ is on a chordless $q_1,q_2$-path $P$, for 
some $q_1$ and $q_2$ in $Q$. Then $q_1$ and $q_2$ are not adjacent. 
The chordless path $P$ must contain two nonadjacent vertices 
$q_1^{\prime}$ and $q_2^{\prime}$ in $\Delta(x)$. But then, by 
Lemma~\ref{lm ordered neighborhoods} the path cannot be chordless, 
unless one of $q_1$ and $q_2$ is in $\Delta(x)$. 
\qed\end{proof}

\bigskip 

The solutions for minimum monophonic sets $Q$ with $x \notin Q$ 
and $y \notin Q$ are easily obtained by a modified input with some 
extremal pair in the solution.  To simplify the description 
somewhat we henceforth assume that there is a minimum monophonic set 
which contains both elements of some extremal pair $\{x,y\}$.  

\bigskip 
  
\begin{definition}
Let $S_1$ and $S_2$ be two minimal separators in a permutation graph $G$. 
Then $S_1$ and $S_2$ are \underline{successional} if 
\begin{enumerate}[\rm (a)]
\item $S_1$ and $S_2$ are parallel, and   
\item $a \in S_1$ and $b\in S_2$ implies $a=b$ or $\{a,b\}\in E$, and   
\item The number of vertices between $S_1$ and $S_2$ is maximal with respect 
of the previous two conditions. 
\end{enumerate}
\end{definition}
 
\bigskip 

In the following we use the notation of lemma~\ref{lm extremal pair}. 
Let $\{x,y\}$ be an extremal pair. 
Let $C$ be a component of $G-J(x,y)-A(x)-A(y)$. 
Assume that the component $C$ is between successional separators 
$S_1$ and $S_2$. By definition of succesional separators and by the nature 
of the permutation diagram, every vertex of $C$ is adjacent to all vertices 
of $S_1$ or to all vertices of $S_2$. 
In other words, 
if some vertex of $S_1$ has no neighbors in $C$ then all vertices of 
$S_2$ are adjacent to all vertices of $C$. 
We say that a vertex $s \in S_1$ is partially adjacent to $C$ 
if $s$ has at least one neighbor and at least one nonneighbor in $C$. 

\bigskip 

We say that a component $C$ is \underline{covered} 
by a set of vertices $Q$, 
if every vertex of $C$ is on a chordless path between vertices in $Q$. 
For a component $C$, let $\xi(C)$ denote the minimum number of 
vertices in $G$ in any cover of $C$, that is, 
\[\xi(C)=\min\;\{\;|Q|\;|\; \text{$C$ is covered by $Q$}\;\}.\]  

\bigskip 

The algorithm recurses on components $C$ that are between succesional 
separators $S_1$ and $S_2$ and tabulates the covers of $C$, by 
distinct elements 
of $S_1$ and $S_2$. 
When $S_1$ is joined to the component, and when it has 
two nonadjacent vertices, only covers of $C$ with at most two 
elements of $S_1$ are of interest. When $S_1$ is a clique, 
joined to the component, then by Lemma~\ref{lm paluga}, no vertex of $S_1$  
is of interest for the cover.  To prove that 
the recursion is polynomial, we show, in the following lemma,  
that a vertex of $S_1$ is {\em partially\/} adjacent 
to at most one component that is between $S_1$ and $S_2$.  

\begin{lemma}
Each vertex of $S_1$ is partially adjacent to at 
most one component $C$ between $S_1$ and $S_2$. 
\end{lemma}
\begin{proof}
Let $s \in S_1$ and assume that $s$ has a neighbor and a 
nonneighbor in $C$. 
Any other component $C^{\prime}$ has all linesegments 
to the left or to the right of all linesegments of $C$. Thus 
either $s$ is not adjacent to any vertex of $C^{\prime}$ or to all 
vertices of $C^{\prime}$. 
\qed\end{proof}

\bigskip 

\begin{lemma}
\label{lm bound S_1}
Let $C$ be a component between successional separators $S_1$ and $S_2$. 
For any cover $Q$ of $C$ there exists an alternative cover $Q^{\prime}$ with 
$|Q^{\prime}| \leq |Q|$ such that $Q^{\prime}$ has at most 4 vertices in $S_1$. 
\end{lemma}
\begin{proof}
Let $C$ be a component between successional separators $S_1$ and $S_2$. 
Let $Q$ be a cover of $C$. Assume that a vertex $c \in C$ 
is on a chordless path between $s \in Q \cap S_1$ and 
$c^{\prime} \in Q \cap C$. 
Then the $c^{\prime},s$-path can be extended from $s$ to either $x$ or $y$. Thus 
we may replace those vertices $s \in Q \cap S$ by the single vertex $x$ or $y$ 
in $Q^{\prime}$.  A similar argument applies to chordless paths  
between two elements of different components $C$ and $C^{\prime}$.  

\medskip 

\noindent 
The other possibility is that a vertex $c \in C$ is on a chordless path between 
two vertices $s_1$ and $s_2$ in $Q \cap S$; that is, $c$ is a 
common neighbor of $s_1$ and $s_2$.  Let $C^{\ast}$ be the set of 
vertices in $C$ that are common neighbors of pairs in $Q \cap S_1$. 
Let $c_1 \in C^{\ast}$ be the vertex  
whose linesegment has an endpoint on the bottom line furthest 
from the scanline $S_1$. Similarly, let $c_2 \in C^{\ast}$ be the vertex 
whose linesegment has an endpoint on the topline furthest 
from the scanline $S_1$. The two pairs of $Q \cap S_1$ that cover 
$c_1$ and $c_2$ cover all 
other vertices of $C^{\ast}$. 
\qed\end{proof}
 
\bigskip 

\begin{theorem}
There exists a polynomial algorithm to compute the monophonic 
number of permutation graphs. 
\end{theorem}
\begin{proof}
We only sketch the detailed, but otherwise standard, dynamic 
programming algorithm. 

\medskip 

\noindent
The proposed algorithm builds a decomposition tree. The root of 
the tree represents the graph $G$, with an extremal pair $\{x,y\}$. 
The children are the components of $G-J(x,y)-A(x)-A(y)$. These 
components are recursively decomposed in subtrees. 

\medskip 

\noindent 
By Lemmas~\ref{lm paluga} and~\ref{lm bound S_1} we may restrict to 
a polynomial number of monophonic covers of each component. 
By dynamic programming, the algorithm groups successive components 
together, and builds a table for the covers of intervals of components. 
Since the number of table entries is polynomial, the algorithm 
runs in polynomial time. 
\qed\end{proof}

\bigskip 

\begin{remark}
We conjecture that a similar algorithm works for $\AT$-free graphs. 
\end{remark}

\end{document}